\documentclass[letterpaper,twocolumn,twoside,aps,superscriptaddress,longbibliography]{revtex4-1}
\usepackage{revsymb,amsmath,amsfonts,amssymb,enumerate,fullpage,amsthm,graphicx,braket,relsize,mathrsfs}
\usepackage{graphicx,color, xcolor}
\usepackage{paralist}
\usepackage{enumitem}

\newtheorem{theo}{Theorem}
\newtheorem{thm}[theo]{Theorem}

\newcommand{\cO}{\mathcal{O}}
\newcommand{\id}{\mathbb{I}}
\newcommand{\tr}[2]{\mathrm{tr}_{#2} \left\{ #1 \right\}}
\newcommand{\trace}[1]{\mathrm{tr}\left\{#1 \right\}}
\newcommand{\proj}[1]{\ket{#1}\!\bra{#1}}
\newcommand{\half}{$\frac{1}{2}$ }

\newcommand{\cM}{\mathcal{M}}

\usepackage{tikz}

\usetikzlibrary{calc,decorations.pathreplacing}
\usetikzlibrary{shapes}
\usetikzlibrary{shapes.arrows}
\usetikzlibrary{calc,decorations.pathreplacing}
\definecolor{darkgreen}{rgb}{0,.5,0}
\definecolor{darkred}{RGB}{230,159,0}
\definecolor{darkblue}{RGB}{86,180,233}

\tikzset{
  mycustomnode/.style={
  single arrow,
  draw, fill,
  inner sep=0.2cm,
  minimum height=4cm,
  single arrow head extend=0.8cm,
  }
}

\begin{document}

\title{{Reference frames which separately store \protect\\ non-commuting conserved quantities}}

\author{Sandu Popescu}
\email{s.popescu@bristol.ac.uk}
\affiliation{H. H. Wills Physics Laboratory, University of Bristol, Tyndall Avenue, Bristol, BS8 1TL, United Kingdom}

\author{Ana Bel\'{e}n Sainz}
\email{ana.sainz@ug.edu.pl}
\affiliation{International Centre for Theory of Quantum Technologies, University of Gda\'nsk, 80-308 Gda\'nsk, Poland}
\affiliation{Perimeter Institute for Theoretical Physics, Waterloo, Ontario, Canada, N2L 2Y5}

\author{Anthony J. Short}
\email{tony.short@bristol.ac.uk}
\affiliation{H. H. Wills Physics Laboratory, University of Bristol, Tyndall Avenue, Bristol, BS8 1TL, United Kingdom}

\author{Andreas Winter}
\email{andreas.winter@uab.cat}
\affiliation{ICREA---Instituci\'o Catalana de Recerca i Estudis Avan\c{c}ats, Pg.~Lluis Companys 23, 08010 Barcelona, Spain}
\affiliation{F\'{\i}sica Te\`{o}rica: Informaci\'{o} i Fen\`{o}mens Qu\`{a}ntics, Departament de F\'{\i}sica, Universitat Aut\`{o}noma de Barcelona, 08193 Bellaterra (Barcelona), Spain}

\date{17 April 2020}

\begin{abstract}
Even in the presence of conservation laws, one can perform arbitrary transformations on a system if given access to a suitable reference frame, since conserved quantities may be exchanged between the system and the frame. Here we explore whether these quantities can be separated into different parts of the reference frame, with each part acting as a `battery' for a distinct quantity. For systems composed of spin-\half particles, we show that the components of angular momentum $S_x$, $S_y$ and $S_z$ (non-commuting conserved quantities) may be separated in this way, and also provide several extensions of this result.  These results also play a key role in the quantum thermodynamics of non-commuting conserved quantities.
\end{abstract}

\maketitle

Conservation laws are amongst the most important and widely used aspects of physics,  greatly restricting the possible transformations that an isolated system can undergo \cite{Wig52, Ara60, Marvian}. However, when a system  is not isolated but allowed to interact with other systems, with conservation laws applying only globally, much greater freedom is possible \cite{AharonovSusskind1, AharonovKaufherr,RobRTA,Ahmadi, Marvian13, QRF18}. The situation is particularly interesting in quantum theory, where different conserved quantities may not commute (such as the different components of angular momentum),  and interference effects are crucial. 

Perhaps surprisingly, it has been shown that any transformation of a quantum system can be implemented, as long as one has access to an appropriate ancillary system \cite{RobRTA,Ahmadi, Marvian13, QRF18}. This additional system plays a dual role of providing a reference frame for the transformation, and acting as a reservoir which can exchange conserved quantities with the system. 

While undergoing its transformation, the system will generally exchange many different conserved quantities with the reference frame. An interesting question is whether the exchange of each of these different conserved quantities can be separated into different parts of the reference system, with each part effectively acting as a `battery' for a specific conserved quantity. When the conserved quantities commute, this is possible \cite{NCCQ}. The main question we raise - and answer - in the present paper, is whether this is possible when the conserved quantities do {\it not} commute. In this situation one might expect that it is impossible to separate the various conserved quantities into different reservoires, because one cannot even measure one of them without disturbing the other. Surprisingly, we show that this is possible. 

As a concrete example, consider rotations of a spin-\half particle in the presence of angular momentum conservation. We would like to be able to localise any changes in the three components of spin $s_x$, $s_y$ and $s_z$ {of the spin-\half particle} within different subsystems in the reference frame {(see Fig.~\ref{thefig})}. In this paper, we show that this is indeed possible. Our approach will generalise to any unitary transformation on any number of spin-\half systems. Furthermore, {as} an interesting example of such a protocol, we show how to completely extract the three components of angular momentum of an unknown spin state into three distinct systems (up to arbitrary precision). Finally, in the Supplementary material, we extend these results for arbitrary conserved quantities of any system of dimension $2^n$ and to conservation of angular momentum under rotations for arbitrary spin.

\begin{figure}
\begin{center}
\begin{tikzpicture}[scale=0.5]
\draw[thick,->] (0,0) -- (-3,-3);
\node at (-3.5,-3) {$s_x$};
\draw[thick,->] (0,0) -- (0,6.5);
\node at (0.6,6.3) {$s_z$};
\draw[thick,->] (0,0) -- (8,0);
\node at (8,0.5) {$s_y$};

\node at (3,4.5) {$\mathbf{s}$};
\node[draw,fill,color=darkblue,shape=circle,scale=.5] (a) at (3,4) {} ;
\draw[thick, color=darkblue,->] (0,0) -- (a);
\draw[dashed, color=darkblue] (0,6) -- (a);
\draw[dashed, color=darkblue] (3,-2) -- (a);
\draw[dashed, color=darkblue] (-2,-2) -- (3,-2);
\draw[dashed, color=darkblue] (5,0) -- (3,-2);

\node at (6,2.5) {$\mathbf{s'}$};
\node[draw,fill,color=darkred,shape=circle,scale=.5] (b) at (6,2) {} ;
\draw[thick, color=darkred,->] (0,0) -- (b);
\draw[dashed, color=darkred] (0,3) -- (b);
\draw[dashed, color=darkred] (6,-1) -- (b);
\draw[dashed, color=darkred] (-1,-1) -- (6,-1);
\draw[dashed, color=darkred] (7,0) -- (6,-1);

\draw[thick, decorate,decoration={brace,amplitude=5pt}] (-2,-2)--(-1,-1);
\node at (-2.5,-1.2) {$\Delta s_x$};
\draw[thick, decorate,decoration={brace,amplitude=5pt}] (5,0)--(7,0);
\node at (6.5,0.7) {$\Delta s_y$};
\draw[thick, decorate,decoration={brace,amplitude=5pt}] (0,3)--(0,6);
\node at (-1,4.5) {$\Delta s_z$};

\draw[thick,fill,rounded corners, color=darkgreen!20!white] (-5,7) rectangle (7,13) ;

\node at (1,12) {$z$-battery};
\draw[thick,rounded corners] (0,8) rectangle (2,11) ;
\draw[thick] (0.3,11) rectangle (0.7,11.3);
\draw[thick] (1.3,11) rectangle (1.7,11.3);

\node[mycustomnode, scale=0.1, rotate=0] at (0.6,9) {};
\node[draw,fill,shape=circle,scale=.5, color=gray!80!white] at (0.5,9) {};
\node[mycustomnode, scale=0.1, rotate=225] at (1.4,8.9) {};
\node[draw,fill,shape=circle,scale=.5, color=gray!80!white] at (1.5,9) {};
\node[mycustomnode, scale=0.1, rotate=225] at (1.4,9.9) {};
\node[draw,fill,shape=circle,scale=.5, color=gray!80!white] at (1.5,10) {};
\node[mycustomnode, scale=0.1, rotate=0] at (0.6,10) {};
\node[draw,fill,shape=circle,scale=.5, color=gray!80!white] at (0.5,10) {};

\node at (-3,12) {$x$-battery};
\draw[thick,rounded corners, shift={(-4,0)}] (0,8) rectangle (2,11) ;
\draw[thick, shift={(-4,0)}] (0.3,11) rectangle (0.7,11.3);
\draw[thick, shift={(-4,0)}] (1.3,11) rectangle (1.7,11.3);

\node[mycustomnode, scale=0.1, rotate=0] at (-3.4,9) {};
\node[draw,fill,shape=circle,scale=.5, color=gray!80!white, shift={(-4,0)}] at (0.5,9) {};
\node[mycustomnode, scale=0.1, rotate=90] at (-2.5,9.1) {};
\node[draw,fill,shape=circle,scale=.5, color=gray!80!white, shift={(-4,0)}] at (1.5,9) {};
\node[mycustomnode, scale=0.1, rotate=90] at (-2.5,10.1) {};
\node[draw,fill,shape=circle,scale=.5, color=gray!80!white, shift={(-4,0)}] at (1.5,10) {};
\node[mycustomnode, scale=0.1, rotate=0] at (-3.4,10) {};
\node[draw,fill,shape=circle,scale=.5, color=gray!80!white, shift={(-4,0)}] at (0.5,10) {};

\node at (5,12) {$y$-battery};
\draw[thick,rounded corners, shift={(4,0)}] (0,8) rectangle (2,11) ;
\draw[thick, shift={(4,0)}] (0.3,11) rectangle (0.7,11.3);
\draw[thick, shift={(4,0)}] (1.3,11) rectangle (1.7,11.3);

\node[mycustomnode, scale=0.1, rotate=90] at (4.5,9.1) {};
\node[draw,fill,shape=circle,scale=.5, color=gray!80!white, shift={(4,0)}] at (0.5,9) {};
\node[mycustomnode, scale=0.1, rotate=225] at (5.4,8.9) {};
\node[draw,fill,shape=circle,scale=.5, color=gray!80!white, shift={(4,0)}] at (1.5,9) {};
\node[mycustomnode, scale=0.1, rotate=225] at (5.4,9.9) {};
\node[draw,fill,shape=circle,scale=.5, color=gray!80!white, shift={(4,0)}] at (1.5,10) {};
\node[mycustomnode, scale=0.1, rotate=90] at (4.5,10.1) {};
\node[draw,fill,shape=circle,scale=.5, color=gray!80!white, shift={(4,0)}] at (0.5,10) {};

\draw[->, thick, color=darkgreen] (-2.5, -0.7) to [in=270, out=90] (-3, 7.9); 
\draw[->, thick, color=darkgreen] (-1.2, 4.9) to [in=270, out=90] (0.5, 7.9); 
\draw[->, thick, color=darkgreen] (6.6, 1) to [in=270, out=90] (5.5, 7.9); 

\end{tikzpicture}
\end{center}
\caption{{A spin-$\tfrac{1}{2}$ system, initially with spin $\mathbf{s}$ 
(color blue online) undergoes a rotation, after which its spin becomes 
$\mathbf{s'}$ (color orange online). During this rotation, its spin 
components in the $x$, $y$ and $z$ direction change by $\Delta s_x$, $\Delta s_y$ and $\Delta s_z$, respectively. This unitary evolution is implemented by acting on the system plus the 
particles of a reference frame (green box, online), composed of spin-$\tfrac{1}{2}$ 
particles that are divided into three groups: (i) the $x$-battery, 
composed of particles in the + eigenstates of $s_y$ or $s_z$, (ii) the $y$-battery, 
composed of particles in the + eigenstates of $s_x$ or $s_z$, and 
(iii) the $z$-battery, composed of particles in the + eigenstates of $s_x$ 
or $s_y$. {Up to arbitrary accuracy, the} particles in the reference 
frame corresponding to the $k$-battery will absorb all and only the change $\Delta s_k$ 
of the evolved spin system.}}\label{thefig}
\end{figure}
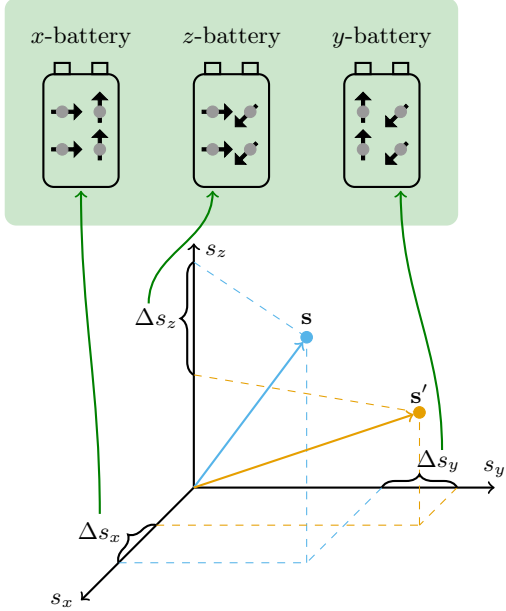

These results address a fundamental aspect of conservation laws  in quantum theory. 
Furthermore, they are of particular importance in quantum thermodynamics, which has 
recently been extended to multiple {non-commuting} conserved 
quantities \cite{vaccaro, barnett, halpern2016, halpern2018, lostaglio2017, NCCQ, jonathan, gour, ito, Manzano, Sparaciari, Yunger}, 
and where batteries for storing conserved quantities are of particular interest.

\textit{Separate batteries for spin-\half systems.---}%
We consider the case of spin-\half particles, and the possibility of separating the different conserved components of angular momentum $s_x$, $s_y$ and $s_z$ {(see Fig.~\ref{thefig})}. We will show that by considering a fixed reference frame composed of multiple spin-\half particles, and interacting with it in a particular way: (i) the total angular momentum of the system and reference frame is conserved, (ii) any unitary transformation can be implemented on the system with arbitrary precision, and (iii) any changes in the average angular momentum components of the system are stored in different parts of the reference system  (up to arbitrarily small correction terms). Essentially, we can think of the reference system as being partitioned into three separate batteries, each of which stores a different component of angular momentum. 

The reference system must: (a) indicate the $x$, $y$ and $z$ directions in order to allow rotations of the system about these directions and (b) do this via a rotationally invariant interaction so that the total angular momentum is conserved. One way to construct a reference frame for the $x$-direction would be to prepare a number of spins all pointing in the $x$-direction (as in \cite{QRF18}). This could be used to implement a rotation of the system about $x$, but each spin in the reference frame would generally accumulate changes in both $s_y$ and $s_z$, thus not separating the conserved quantities. The key intuition behind our approach is that there is an alternative way to define the $x$-direction. Instead of aligning each spin in the reference frame along the $x$-direction, we prepare pairs of spins pointing in the $y$ and $z$ direction, and implement a cross product via the interaction. In this way, each spin in the reference frame accumulates changes in only one component of spin, and thus the different conserved quantities can be separated. 

Our result is formalised in the following theorem. Using lower case $s$ for the spin of individual spin-\half particles and capital $S$ for the spin of systems composed by many spins, we have the following. Let our system be a spin-\half particle, whose spin is denoted by $\textbf{s}$, with components $\{s_k\}_{k=x,y,z}$.  Let the total spin of the $j$-part of the reference frame (i.e. the part which is intended to store spin in the $j$-direction) be denoted by $\textbf{S}^{(j)}$, {$j = x,y,z$}, and has components $\{S_k^{(j)}\}_{k:x,y,z}$; and finally, the total spin operator of the system and frame is denoted by ${\bf S}^{\text{tot}}=\textbf{s}+\textbf{S}^{(x)}+\textbf{S}^{(y)}+\textbf{S}^{(z)}$, and its components given by $S^{\text{tot}}_k = s_k + S_k^{(x)}+S_k^{(y)}+S_k^{(z)}$, $k = x,y,z$. 

\begin{thm} 
\label{thm:spin} 
Let the system ${\mathcal S}$ be a {spin-$\frac12$} particle. Then for every $\epsilon>0$ and $\delta> 0$, there exists a reference frame $\cal R$  (composed of a large number of spin-half particles) with a fixed state $\rho_R = \rho^{(x)}_R \otimes \rho^{(y)}_R \otimes \rho^{(z)}_R$,  where $\rho^{(k)}_R$ is the state of the $k$-part of the reference system, such that for every unitary $U_S$ on the system there exists a joint unitary $V$  on the system and reference frame with the following properties: 
\begin{itemize}
\item \emph{Conservation:} $V$ conserves all components of total angular momentum $\textbf{S}^{\text{tot}}$, {i.e.,} $[V, \textbf{S}^{\text{tot}}] = 0$. 
\item \emph{Accuracy:} $V$ effectively implements $U_S$ on the system
{with precision $\epsilon$, i.e., for any initial system state $\rho_S$,}
\begin{equation}\label{eq:thefirst}
  \phantom{==}
  \left\| \tr{V \, \rho_S \otimes \rho_R\, V^\dagger }{R} - U_S \rho_S U_S^\dagger \right\|_1 \leq \epsilon. 
\end{equation}
\item \emph{Separation:} Each component of  angular momentum of the system is exchanged only with the corresponding part of the reference system, up to {precision $\delta$:}
\begin{align} \label{eq:second}
\left| \Delta s_j  + \Delta S_j^{(j)} \right| &\leq \delta, \\
\left| \Delta S_j^{(k)} \right| & \leq \delta  \quad \text{if } j \neq k, \label{eq:third}
\end{align} 
where $\Delta s_j$ is equal to the change in the average angular momentum of the system in the $j$-direction, and  $\Delta S_j^{(k)}$ is equal to the change in the average angular momentum of the $k$-part of the reference system in the $j$-direction.
\end{itemize}
\end{thm} 

\begin{proof} 
The proof strategy is to first show how to implement a small rotation of the system around the $x$-axis, by using the system plus two spins in the reference frame, one polarised in the $y$-direction and one in the $z$-direction. We then extend this to small rotation about an arbitrary axis, by doing subsequent small rotations about the $x$, $y$ and $z$-direction. Finally we repeat this many times to build up a general rotation the system. 

The key is to consider the following operator acting on three spin-\half particles 
\begin{equation} \label{eq:theT}
T =\mathbf{s} \cdot   (\mathbf{s}' \times \mathbf{s}'') = \sum_{j,k,\ell \in \{x,y,z\}} \epsilon_{jk\ell} \,s_j  s'_k s''_\ell\,,
\end{equation} 
{where {$\mathbf{s}$, $\mathbf{s}'$ and $\mathbf{s}''$} are spin operators of the three particles.
In the {right-hand side} of Eq.~\eqref{eq:theT}, the} sub-index $j$ in the spin operators denotes the spin component in the $j$ direction, and $\epsilon_{jk\ell} $ the Levi-Civita tensor.  For simplicity we have set $\hbar=1$, so each spin operator is equal to half the corresponding Pauli operator.

The operator $T$ is invariant under rotations, as it is a ``scalar'', being defined as a 
dot product of two vectors, $\mathbf{s} $ and  $(\mathbf{s}' \times \mathbf{s}'')$. 
We can now construct a unitary interaction between three qubits given by 
\begin{align}
\label{eq:spinU}
  V_\alpha = \exp\{ -i \tfrac{4 \alpha}{N} T \}.
\end{align}

Since $T$ is invariant under rotations, it commutes with the total spin 
$\textbf{S}^{\text{tot}} = \mathbf{s} + \mathbf{s}' + \mathbf{s}''$. In particular this means 
that  also $V_{\alpha}$ preserves the conserved quantities $S_x$, $S_y$ and $S_z$, 
and hence satisfies  $[V_\alpha, \textbf{S}] = 0$ (see Appendix).
 
Let us define $\tau_j$ as the density matrix of a spin-\half particle pointing in the $j$ 
direction (i.e. $\tau_j={I\over 2}+s_j$); for example $\tau_z = \proj{\uparrow_z}={I\over 2}+s_z$, 

Suppose that we want to implement a small rotation of the system about the $x$-direction, given by $U_{\alpha,x} = \exp\left(-i \frac{\alpha}{N} s_x \right)$. To do this we prepare the spins $\mathbf{s}'$ and $\mathbf{s}'' $ in the density matrix $\tau_y' \otimes \tau_z''$ and act with $V_\alpha$ on the system and these two spins. 

We show in the Appendix that
\begin{align}  \label{eq:single-spin} 
\nonumber \tr{ V_{\alpha} \, \rho_S \otimes \tau_y' \otimes \tau''_z \, V_{\alpha}^\dagger }{R} &=  
U_{\alpha, x} \,\rho_S \, U_{\alpha, x}^{\dagger} \\ &+ \cO\left( \tfrac{1}{N^2} \right) .
\end{align} 
In this way, the state  $\tau_y' \otimes \tau_z''$ defines a reference frame for the direction $x$.

We now consider how the angular momentum in the reference system changes under this transformation. 

\begin{align} 
  \label{eq:spindelta1} 
  \Delta s'_z  &= - \frac{\alpha}{N} \trace{ s_y \rho_S} +  \cO\left( \tfrac{1}{N^2} \right), \\
  \label{eq:spindelta2} 
  \Delta s''_y &= \frac{\alpha}{N}  \trace{ s_z \rho_S} +  \cO\left( \tfrac{1}{N^2} \right), 
\end{align}
while all other components of the reference spins are left unchanged, up to 
$\cO\left( \tfrac{1}{N^2} \right)$.
{These} equations are proven in the Appendix.
Hence to leading order, the first reference system only picks up $z$-spin, and the second reference system only picks up $y$-spin.
{Therefore,  in the context of the overall reference frame, $\mathbf{s}'$ belongs to the $z$-part, and $\mathbf{s}''$ to its $y$-part.}  

As total angular momentum is conserved, it follows that the change in the system's angular momentum obeys 
\begin{align} \label{eq:work_sep}
\Delta s_x &=  \cO\left( \tfrac{1}{N^2} \right), \quad 
\Delta s_y = -  \Delta s''_y + \cO\left( \tfrac{1}{N^2} \right), \nonumber \\
\Delta s_z &= -  \Delta s'_z + \cO\left( \tfrac{1}{N^2} \right).
\end{align}

Similarly, due to the cyclic symmetry of the spin operators, we can generate a small rotation of the system about the $y$-direction or $z$-direction by acting with $V_{\alpha}$ on the system and a reference frame  in the state  $\tau_z' \otimes \tau_x''$ or $\tau_x'\otimes \tau_y''$ respectively.  In each case,  the two components of angular momentum {that} change (perpendicular to the axis of rotation) are separated into the two different reference spins. For example, when performing a small $y$-rotation with the frame $\tau_z'\otimes \tau_x''$, to first order in $\frac{1}{N}$ only the {$x$-spin} of the first reference particle and the $z$-spin of the second reference particle are modified. Intuitively, neither reference system changes its spin-component parallel to the axis of rotation of the system, and each reference spin does not change its spin-component parallel to the direction it was originally pointing in, as it is  maximal in this direction and we are considering only first order changes. 
  
A small rotation about a general axis, given by the unitary $U_H = \exp\left( -i\frac{H}{N} \right) $ where   $H = \sum_{k=x,y,z} \alpha_k \, s_k$ can be generated by performing three subsequent small rotations, around the $x$, $y$ and $z$-directions as described above. In particular, 
we use 6 spin-\half particles, in the state 
\begin{equation} 
{\tau_R= \tau_y^{(z)} \otimes \tau_z^{(y)} \otimes \tau_z^{(x)} \otimes \tau_x^{(z)} \otimes   \tau_x^{(y)} \otimes \tau_y^{(x)}.}
\end{equation} 
where the superscripts denote which part of the global reference frame the spins are in, i.e. which angular momentum component they will store.
We then implement $U_H $ by first applying $V_{\alpha_x}$ to the system and the first two reference spins, then $V_{\alpha_y}$ to the system and the next two reference spins, then $V_{\alpha_z}$ to the system and the last two reference spins. Following a similar approach to {Eq.~\eqref{eq:single-spin},} we thereby obtain 
\begin{multline} 
\tr{V_{\alpha_z} V_{\alpha_y} V_{\alpha_x}\, \rho_S \otimes \tau_R  \, V_{\alpha_x}^\dagger V_{\alpha_y}^\dagger V_{\alpha_z}^\dagger}{R}  \\
= U_H \,\rho_S \, U_H^{\dagger} + \cO\left( \tfrac{1}{N^2} \right).
\end{multline}

Finally, we iterate this procedure $N$ times using a new set of six reference spins in the state $\tau_R$ each time.  The  overall reference frame is therefore $\rho_R= \tau_R^{\otimes N}$, which consists of 6$N$ spins (2$N$ in each of the three parts).

In this way  we will approximately implement the desired transformation $U_S =( U_H)^N=\exp\left( -iH \right) $, which is (up to a global phase)  the most general unitary transformation on a spin-\half system.

Defining the full sequence of  transformations by $V$, we  note that as this is a sequence of $V_{\alpha}$ transformations, each of which conserve the three components of angular momentum, $V$ will satisfy the \emph{conservation} property $[V, \textbf{S}^{\text{tot}}] = 0$.

The  error in implementing $U_S$ is bounded by the sum of the errors from each step, giving a total error of $N \cO \left( \frac{1}{N^2}\right) = \cO \left( \frac{1}{N} \right)$~\cite{QRF18}. It follows that
\begin{equation} \label{eq:11}
  \left\| \tr{V \, \rho_S \otimes \rho_R\, V^\dagger }{R} - U_S \rho_S U_S^\dagger \right\|_1 \leq  \cO \left( \tfrac{1}{N} \right),
\end{equation} 
which proves the \emph{accuracy} property in  {Eq.~\eqref{eq:thefirst}} by suitably choosing a sufficiently large $N$ according to the value of $\epsilon$. 

To prove the \emph{separation} property, 
we use {Eq.~\eqref{eq:work_sep}} and its equivalents for  $y$ and $z$ rotations, together with the fact that there are only $2N$ spins in each part of the reference system, imply that 
\begin{align} \label{eq:12}
\left| \Delta s_j  + \Delta s_j^{(j)} \right| &\leq \cO \left( \tfrac{1}{N} \right),  \\ \label{eq:13}
\left| \Delta s_j^{(k)} \right|               &\leq \cO \left( \tfrac{1}{N} \right)  \quad \text{if } j \neq k.
\end{align} 
For any $\delta$, we can therefore choose a sufficiently large $N$ such that {Eqs.~\eqref{eq:second}} and \eqref{eq:third} {hold.} Explicit bounds for the $\cO \left( \frac{1}{N} \right)$ and $\cO \left( \frac{1}{N^2} \right)$  terms in this section may be found in the Appendix.
\end{proof} 

These results can be  extended to systems composed of any number of spin-\half particles, {as we argue in the following. First, the} above proof shows that we can implement any unitary on a single spin. {Then, we} can also implement interactions between two spins inside the system via an interacting unitary such as $\sqrt{SWAP}$ or $e^{-i \theta\, \mathbf{s}^{1} \cdot  \mathbf{s}^{2}}$. These commute with  all extensive conserved quantities and do not require the use of a reference system. Thinking of our spin-\half systems as qubits, we know that the ability to perform all single-qubit unitaries plus any particular interacting two-qubit unitary is computationally {universal \cite{DupontDupond}.} Hence we can construct a circuit to approximately implement any unitary transformation on any number of spin-half systems, whilst storing any changes to angular momentum in different batteries.

{\textit{Extracting the angular momentum components of an unknown spin.---}%
An interesting possibility enabled by the  above procedure is to take an unknown spin-\half state with average spin $\langle {\bf s}\rangle = ( \langle s_x\rangle, \langle s_y\rangle, \langle s_z\rangle)$, and three other systems, and  completely extract the different components of spin into the three systems (up to arbitrary accuracy). That is, we can perform a unitary transformation such that, up to arbitrary accuracy,  
the spin-\half particle finally has average spin zero, and the average spin of the three systems has increased by $(\langle s_x\rangle, 0, 0 )$, $(0, \langle s_y\rangle, 0)$, and $(0, 0, \langle s_z\rangle)$ respectively.

To do this, we use the $x$, $y$ and $z$-parts of our reference frame as the three systems, and include two ancillary spin-\half particles, each in a maximally mixed state, in one of the systems. We then perform a unitary $V$ on the entire state which is given by the circuit construction above, and which approximately implements a unitary 
\begin{align} 
U  = \sum_{n,m=0}^{1} X^n Z^m \otimes \proj{n} \otimes \proj{m}  
\end{align}
on the initial spin  and the two maximally mixed spins, where  $X$ and $Z$ are the Pauli unitaries
on the system. The net effect of $U$ is to completely decohere the state of the spin-\half particle and transform it into the maximally mixed state \cite{tonyFn}. Because the ancillary spins are also left in maximally mixed states, all of the average angular momentum in the initial state of the spin must have been transferred to  the three components  of the reference frame, and hence to the three desired systems. }

\textit{Conclusions.---}%
We have shown that it is possible to perform an arbitrary unitary transformation on any number of spin-\half particles whilst respecting angular momentum conservation, in such a way that any changes in the three components of angular momentum are separated into different `batteries'. Any errors in this procedure can be made arbitrarily small by making these batteries sufficiently large. 

Importantly, the use of the cross product technique is more than a simple technical development. It has, we believe, a deep conceptual meaning. It shows that in the realm of quantum mechanics one can build frames of reference in various ways, all equally good for acting as references (for specifying a direction in our case), but which have fundamentally different properties. We have established this result in the particular case of specifying a direction, but we expect this to be a general property of all quantum frames.

{Our results also allow one to completely extract the different components of angular momentum of an unknown spin state into distinct systems (up to arbitrary precision). It would be interesting to investigate the ultimate limits of such a procedure, such as whether it can be made exact, and the smallest possible implementation.}

The study of quantum thermodynamics has recently been extended to other conserved quantities besides energy, {in particular to non-commuting conserved quantities,} and this result allows one to consider explicit batteries for the angular momentum. The fact that the different components of angular momentum can be separated in this way is particularly surprising given that they do not commute{, and measurements of one component would disturb the others}. 

The protocols we have described for spin-$\tfrac{1}{2}$ systems may be generalised to higher spins,
when the objective is to implement a spatial rotation under total angular momentum conservation. 
However, such a higher-dimensional quantum system has many more observables, each of which could be a 
potential conserved quantity, up to the maximum of $d^2-1$ for a system with $d$-dimensional 
Hilbert space (for instance $d=2\sigma+1$ for a spin-$\sigma$ particle). 
One may ask whether it is possible to construct a reference frame, with a part for each of 
of the observables to be conserved, such that an arbitrary transformation on the
system can be implemented by an operation that globally obeys all potential conservation laws,
while the change of the $k$-th quantity on the system is offset by the 
same change of only the $k$-th part of the frame, all to arbitrary precision.
In the case that the dimension is a power of $2$, we show in the Supplemental 
Material that a complete set of conserved quantities can be constructed for which 
this is possible, by an adaptation of our spin-\half methods. The question
remains open, however, for other dimensions and general sets of quantum observables.

\textit{Acknowledgements.---}%
SP acknowledges support from The Institute 
for Theoretical Studies, ETH Zurich. 
ABS acknowledges support by the Foundation for Polish Science 
(IRAP project, ICTQT, contract no. 2018/MAB/5, co-financed by EU within 
Smart Growth Operational Programme). 
AW acknowledges support by the Spanish MINECO (project FIS2016-86681-P) 
with the support of FEDER funds, and the Generalitat de Catalunya (project 2017-SGR-1127).
This research was partially supported by Perimeter Institute for Theoretical Physics. 
Research at Perimeter Institute is supported by the Government of Canada through the 
Department of Innovation, Science and Economic Development Canada and by the 
Province of Ontario through the Ministry of Research, Innovation and Science.

\appendix

\onecolumngrid

\section{Explicit bounds for  $\cO\left( \frac{1}{N^2} \right)$ and 
         $\cO\left( \frac{1}{N} \right)$ terms for spin-\half systems}
\label{ap:spindec}
Here we provide more detailed technical proofs giving explicit bounds for the  
$\cO\left( \frac{1}{N^2} \right)$ and $\cO\left( \frac{1}{N} \right)$ terms in 
our bounds for spin-\half systems, obtained using similar techniques to those 
introduced in Ref.~\cite{QRF18}.

\subsection{Proof of {Eq.~\eqref{eq:single-spin}}}
{First notice that
\begin{align}  
\nonumber \tr{ V_{\alpha} \, \rho_S \otimes \tau_y^{\prime} \otimes \tau^{\prime\prime}_z \, V_{\alpha}^\dagger }{R} 
    &= \rho_S - i \frac{4\alpha}{N} \tr{ [ T, \rho_S \otimes \tau_y^{\prime}  \otimes \tau_z^{\prime\prime} ]}{R} +  \cO\left( \tfrac{1}{N^2} \right)\,, \nonumber \\
    & =  \rho_S - i \frac{4\alpha}{N} \!\!\! \sum_{j,k,\ell \in \{x,y,z\}}\hspace{-.5cm} \epsilon_{jk\ell}  [ s_j, \rho_S] \, \trace{\!s^{\prime}_k \tau^{\prime}_y\!} \trace{\! s^{\prime\prime}_{\ell} \tau^{\prime\prime}_z\!}  +  \cO\left( \tfrac{1}{N^2} \right)\,, \nonumber \\
    & =  \rho_S - i \frac{\alpha}{N}  [ s_x, \rho_S] +  \cO\left( \tfrac{1}{N^2} \right)\,, \nonumber \\
    & = U_{\alpha, x} \,\rho_S \, U_{\alpha, x}^{\dagger} + \cO\left( \tfrac{1}{N^2} \right) .
\end{align} 
What we need to do now is to provide a rigorous computation for the 
$\cO\left( \tfrac{1}{N^2} \right)$ term. To do this, we must first show that}
\begin{equation} 
  \label{eq:xi}
  \xi \equiv \left\| \tr{ V_{\alpha} \, \rho_S \otimes \tau_y^{\prime} \otimes \tau^{\prime\prime}_z \, V_{\alpha}^{\dagger} }{R}  - ( \rho_S - i \frac{4\alpha}{N} \tr{ [ T, \rho_S \otimes \tau_y^{\prime}  \otimes \tau_z^{\prime\prime} ]) }{R} \right\|_1 
   \leq \cO\left( \frac{1}{N^2} \right)
\end{equation} 
where $V_\alpha = \exp\{ -i \tfrac{4\alpha}{N} T \}$. Expanding the exponentials we obtain 
\begin{align}  \label{eq:16}
\xi   &= \left\| \sum_{n=2}^\infty  \sum_{k=0}^n \tr{ \frac{(-i \frac{4\alpha}{N}\, T)^k}{k!} \rho_S \otimes \tau_y^{\prime} \otimes \tau^{\prime\prime}_z   \frac{(i \frac{4\alpha}{N} \, T)^{n-k}}{(n-k)!}}{R} \right\|_1 \nonumber \\
      &\leq  \sum_{n=2}^\infty \left(  \frac{4\alpha}{N}\right)^n \sum_{k=0}^n \frac{1}{k! (n-k)!} \, \left\| T^k   \rho_S \otimes \tau_y^{\prime} \otimes \tau^{\prime\prime}_z  \, \, T^{n-k}  \right\|_1 \nonumber \\
      &\leq  \sum_{n=2}^\infty \left(  \frac{4\alpha}{N}\right)^n \sum_{k=0}^n \frac{1}{k! (n-k)!}  \| T^k \| \, \| \rho_S \otimes \tau_y^{\prime} \otimes \tau^{\prime\prime}_z \|_1 \,   \| T^{n-k} \| \nonumber \\
      &\leq  \sum_{n=2}^\infty \left(  \frac{4\alpha}{N}\right)^n \sum_{k=0}^n \frac{1}{k! (n-k)!}  \|T\|^n \nonumber \\
      &\leq \sum_{n=2}^{\infty} \frac{1}{n!} \left(  \frac{3\alpha}{N}\right)^n \sum_{k=0}^n \frac{n!}{k! (n-k)!} \nonumber \\
      &=  \sum_{n=2}^{\infty} \frac{1}{n!} \left(  \frac{6\alpha}{N}\right)^n  \nonumber \\
      &= 36 \left(  \frac{\alpha}{N}\right)^2  \sum_{n=2}^{\infty} \frac{1}{n!} \left(  \frac{6\alpha}{N}\right)^{n-2}. 
\end{align} 
where in the fifth line we have used  the fact that 
$T=\sum_{j,k,\ell \in \{x,y,z\}} \epsilon_{jk\ell} \,s_j  \otimes s^{\prime}_k \otimes s^{\prime\prime}_\ell$ 
has 6 non-zero terms, each of which has operator norm 
$\frac{1}{8}$ (as $\| s_j\| = \frac{1}{2}$), giving 
$\|T\| \leq \frac{3}{4}$. Hence for $N \geq 6 \alpha$, 
\begin{equation}  
  \label{eq:above} 
  \xi \leq 36 \, (e-2) \left(  \frac{\alpha}{N}\right)^2. 
\end{equation} 
On the other hand, applying {Eq.~(D5)} of Ref.~\cite{QRF18} to our setup yields
\begin{equation}\label{eq:23} 
\left\| U_{\alpha, x} \rho_S U_{\alpha, x}^{\dagger} - \left(\rho_S - i \frac{\alpha}{ N} [ s_x, \rho_S] \right) \right\|_1 \leq 4 (e-2) \left(  \frac{\alpha}{N}\right)^2,
\end{equation} 
where $U_{\alpha,x} = \exp\left(-i \frac{\alpha}{N} s_x \right)$. 

Combining {Eqs.~\eqref{eq:above}} and \eqref{eq:23} through the triangle inequality, we obtain 
\begin{equation} 
\left\| \tr{ V_{\alpha} \, \rho_S \otimes \tau_y^{\prime} \otimes \tau^{\prime\prime}_z \, V_{\alpha}^{\dagger} }{R} - U_{\alpha, x} \rho_S U_{\alpha, x}^{\dagger} \right\|_1 \leq 40 (e-2) \left(  \frac{\alpha}{N}\right)^2,
\end{equation} 
where we have used the fact that $\tr{ [ T, \rho_S \otimes \tau_y^{\prime} \otimes \tau_z^{\prime\prime}]}{R} = \frac{1}{4} [ s_x, \rho_S]$.

\subsection{Proof of Eqs.~\eqref{eq:spindelta1} and \eqref{eq:spindelta2}}
\label{ap:A2}
{Here we will discuss the case of Eq.~\eqref{eq:spindelta1}, 
since Eq.~\eqref{eq:spindelta2} follows similarly. First notice that 
\begin{align} 
  \Delta s^{\prime}_i &= \trace{(\openone \otimes s_i^{\prime} \otimes \openone) V_{\alpha} \, (\rho_S \otimes \tau^{\prime}_y\otimes \tau^{\prime\prime}_z) \, V_{\alpha}^\dagger } - \trace{s_i^{\prime} \tau_y^{\prime}}, \nonumber \\
&=  - i \frac{4\alpha}{N} \sum_{j,k,\ell \in \{x,y,z\}} \epsilon_{jk\ell} \trace{ s_j \rho_S} \trace{s_i^{\prime} [s^{\prime}_k, \tau^{\prime}_y]} \trace{s^{\prime\prime}_\ell \tau^{\prime\prime}_z}+  \cO\left( \tfrac{1}{N^2} \right),\nonumber \\  
&=  - \frac{\alpha}{N} \delta_{i,z}  \trace{ s_y \rho_S} +  \cO\left( \tfrac{1}{N^2} \right),  \label{eq:thenewone}
\end{align}
where in the last step we have noted that  the expression is only non-zero when $\ell=z$, $k=x$  and $j=y$, and used the spin commutation relation $[s_a, s_b]=i \epsilon_{abc} s_c$.

 Now, applying similar arguments to those in  {Eqs.~\eqref{eq:xi}}--\eqref{eq:above}  to the changes in spin given by {Eq.~\eqref{eq:thenewone}} we obtain }
\begin{align} 
\left| \Delta s_i^{\prime} + \frac{\alpha}{N} \delta_{i,z}  \trace{ s_y \rho_S} \right| \leq 18 \, (e-2) \left(  \frac{\alpha}{N}\right)^2, \\
\left| \Delta s_i^{\prime\prime}- \frac{\alpha}{N} \delta_{i,y}  \trace{ s_z \rho_S} \right| \leq  18 \, (e-2) \left(  \frac{\alpha}{N}\right)^2.
\end{align}

\subsection{Proofs of Eqs.~\eqref{eq:11}, \eqref{eq:12} and \eqref{eq:13}}
\label{ap:seq}
Let us start by proving {Eq.~\eqref{eq:11}.}  
When considering a general small rotation, we must obtain a bound on  
\begin{align} 
\zeta \equiv \left\| \tr{ V_{\alpha_z} V_{\alpha_y} V_{\alpha_x} ( \rho_S \otimes \tau_R) \, V_{\alpha_x}^{\dagger} V_{\alpha_y}^{\dagger} V_{\alpha_z}^{\dagger}}{R} - (\rho_S -i \tfrac{1}{N} \, [H, \rho_S]) \right\|_1 .
\end{align} 
Following a similar argument to that below {Eq.~(D7)} of Ref.~\cite{QRF18}, expanding the unitaries and collecting terms in $\left( \frac{1}{N}\right)^n$ for $n \geq 2$ we obtain at most $6^n$ contributions (as each of the 6 unitaries could provide each power of $\frac{1}{N}$), with a constant coefficient upper bounded by $(4\alpha_{\max})^n$  where $\alpha_{\max} \leq \pi$. Each term also contains the trace norm of an operator with $n$ copies of $T$, acting on different parties, distributed either before or after the initial state. Following the argument above {Eq.~\eqref{eq:above},} such a term is upper bounded by $\left(\frac{3}{4}\right)^n$. Combining all of these observations we obtain 
\begin{equation} 
\zeta \leq \sum_{n=2}^{\infty} \left( \frac{18 \pi}{N}\right)^n.
\end{equation}  
For $N \geq 36 \pi$ we therefore obtain $\zeta \leq 648 \pi^2 \left( \frac{1}{N} \right)^2$. Combining this with {Eq.~(D11)} of Ref.~\cite{QRF18}, for $D=3$, we obtain 
\begin{align} 
\left\| V_{\alpha_z} V_{\alpha_y} V_{\alpha_x} (\rho_S \otimes \tau_R) V_{\alpha_x}^{\dagger} V_{\alpha_y}^{\dagger} V_{\alpha_z}^{\dagger}{R} - U_H \rho_S U_H^{\dagger} ) \right\|_1  \leq \left(648 + 16(e-2)\right) \frac{\pi^2}{N^2}
\end{align} 
for sufficiently large $N$. Iterating this transformation $N$ times and using the inductive argument in Appendix C of Ref.~\cite{QRF18}, gives 
\begin{equation} 
\left\| \tr{V \, \rho_S \otimes \rho_R\, V^\dagger }{R} - U_S \rho_S U_S^\dagger \right\|_1 \leq \left(648 + 16(e-2)\right) \frac{\pi^2}{N}.
\end{equation} 
This completes the proof of {Eq.~\eqref{eq:11}.}

To now prove {Eqs.~\eqref{eq:12}} and \eqref{eq:13}, 
we can apply a similar argument as above. This leads to
\begin{align} 
\left| \Delta s_j  + \Delta s_j^{(j)} \right| &\leq 648 \pi^2 \left( \frac{1}{N} \right) , \\
\left| \Delta s_j^{(k)} \right| & \leq  324 \pi^2 \left( \frac{1}{N} \right) ,
\end{align} 
from which the argument follows.

\section{Separating conserved quantities for higher dimensional systems}
\label{ap:gen}
We will now consider whether we can separate changes  
in conserved quantities for higher dimensional systems, with dimension $d>2$.

\subsection{{Introductory thoughts on systems of dimension $d>2$ and spin conservation laws}}
{Let us first consider angular momentum conservation of spin $s$ systems, with dimension $d=2s+1$. In this case, we can use essentially the same procedure as {in the main text} {(now taking $\tau_j$ as the eigenstate of $s_j$ with maximum spin in the $j$-direction, and $  V_\alpha = \exp\{ -i \tfrac{\alpha}{s^2 N}\}$)} to approximately implement any single system rotation of the form $e^{-i \theta \, \mathbf{n} \cdot \mathbf{s}}$ on a spin-$s$ system, whilst conserving the three components of total angular momentum $S_x$, $S_y$ and $S_z$. Any changes in angular momentum will  be separated into three different `batteries' as before. Furthermore, we could again implement  interacting unitaries between different spins such as  $e^{-i \theta\, \mathbf{s}^{(1)} \cdot  \mathbf{s}^{(2)}}$, as this is rotationally invariant. 

However, unlike in the case of spin-\half particles, spatial rotations of the form $e^{-i \theta \, \mathbf{n} \cdot \mathbf{s}}$ no longer represent the complete set of local unitary transformations for spin-$s$ particles. It would be interesting to explore what additional  transformations are required in order to give a universal gate set, and whether these can be implemented in such a way as to localise any changes to angular momentum. 
 
Moving from angular momentum conservation to more general conservation laws, {an interesting question} is {whether} we can construct a complete basis of conserved quantities, such that arbitrary unitary transformations of a system can be performed whilst separating the changes in all of these conserved quantities into different `batteries'. 
{Later in this section} { we generalise {the proofs} to show that this is possible when the dimension of the system is $2^n$.}
}

\subsection{{Generalisation to arbitrary dimension: what we can and may not do (yet)}}

Here we present a generalisation of the reference frame {defined in the main text for spin-$\frac12$ systems,} and of the operator $T$ of {Eq.~\eqref{eq:theT},} to higher dimensional systems. In particular, we give sufficient conditions to be able to construct a complete basis of extensive conserved quantities for such systems, and a reference frame that allows us to perform arbitrary unitary transformations of the system whilst storing any changes in these conserved quantities in different batteries (up to arbitrarily small corrections).

Consider a system $S$, of dimension $d>2$. In this case there exist $K = d^2 -1$ possible linearly independent  conserved quantities, where without loss of {generality} we do not include the scalar. In what follows, we will consider the case in which there exists a particular choice for these conserved quantities given by K Hermitian operators  $\cM=\{O_k\}_{k=0\ldots K-1}$ with the following properties
\begin{enumerate}[label=(\roman*)] 
\item Traceless: $\tr{O_k}{}=0$ for all $k$ 
\item Orthogonal: $\tr{O_k O_\ell}{} =0$ if $k \neq \ell$
\item  Closed under commutation: For each $k,\ell$ either $[O_k,O_\ell]=0$, or there exists an $m$ such that  $[O_k,O_\ell] \propto O_m$.
\end{enumerate}  
Together with the identity these operators form an orthogonal operator basis for the $d$-dimensional system. In (iii), note that the constant of proportionality could depend on $k$ and $\ell$, and that $m \neq \ell$ as $\tr{ [O_k,O_\ell] O_\ell}{}= \tr{O_k [O_\ell,O_\ell]}{}=0$ (and similarly $m \neq k$). 

First let us consider how to perform the small unitary transformation $U_{\alpha, 0} = \exp(-i \frac{\alpha}{N} \, O_0)$. To do this, we first prepare a reference frame consisting of $D = K-1$ particles of the same type as the system, initialised in the state 
\begin{align}
\rho_R = \rho^{(1)} \otimes \ldots \otimes \rho^{(D)}\,,
\end{align}
where $\rho^{(k)} = \frac{1}{d} \left( \id + \frac{1}{\| O_k \|} \, O_k \right)$. Given that the operators $O_k$ are traceless and orthogonal, it follows that $\trace{O_k \rho^{(r)}} = \eta_k \delta_{k,r}$ where $\eta_k = \tr{O_k^2}{}/(d \| O_k\|)$.

A generalisation of the interaction $T$ of {Eq.~\eqref{eq:theT}} to systems of arbitrary dimensions is
\begin{align}\label{eq:genT}
T = \sum_{a, a_1, \ldots, a_D} f_{aa_1\ldots a_D} \, O_a \otimes O_{a_1} \otimes \ldots \otimes O_{a_D}\,,
\end{align}
where $f_{aa_1\ldots a_D} = 0$ unless all the sub-indices are different, in which case  $f_{aa_1\ldots a_D} = \pm 1$, where the sign indicates if the sub-indices are an even or odd permutation of $0,1, \ldots D$.  

Notice that for $d=2$, $f$ is just the antisymmetric symbol $\epsilon_{ijk}$. In addition, when $d=2$ and $\cM = \{s_k\}_{k=x,y,z}$, we recover the interaction defined in {Eq.~\eqref{eq:theT}} and the batteries defined in {the main text for spin-$\frac12$ systems.}

The main requirement that $T$ should satisfy is $[T,\mathbf{C}]=0$ for any extensive conserved quantity $\mathbf{C}$. We will show this in the next subsection, but for now we focus on using it to perform transformations of the system. The operator $T$  allows us to define a global unitary interaction among $K$ systems given by 
\begin{align}\label{eq:theint}
V =  \exp \left(-i \frac{1}{\eta_1 \ldots \eta_D} \frac{\alpha}{N} T \right)\,,
\end{align}
similarly to {Eq.~\eqref{eq:spinU}.}

The effective action of this global unitary $V$ on our system $A$ is given by $\tr{V \, \rho_S \otimes \rho_R \, V^\dagger}{R}$. Computing this trace, we find that only the term in $T$ containing $f_{012 \ldots D}$ gives a non-zero contribution to first order in $\frac{1}{N}$, leading to 
\begin{align} 
 \tr{V \, \rho_S \otimes \rho_R \, V^\dagger}{R} &=  \rho_S -i\frac{\alpha}{N}  \, [O_0,\rho_S] + \mathcal{O}\left(\frac{1}{N^2}\right)\nonumber \\ 
& =U_{\alpha, 0}  \rho_S U_{\alpha, 0}^{\dagger}  + \mathcal{O}\left(\frac{1}{N^2}\right), 
\end{align}

Similarly we can perform the small unitary transformation $U_{\alpha, r} = \exp(-i \frac{\alpha}{N} \, O_r)$,  generated by an arbitrary operator $O_r$ by preparing the reference frame such that $\rho^{(k)} = \frac{\id}{d} + c_{(k+r) \textrm{mod}  D} \, O_{(k+r) \textrm{mod} D}$ and applying $V$, such that the only non-zero term in the trace corresponds to $f_{r\,(r+1)\ldots D\,0\,1\ldots (r-1)}.$ As in the main paper, by considering a sequence of small transformations  generated by each operator $O_r$ we can implement an arbitrary small transformation $\exp(-i \frac{H}{N})$. By iterating this process $N$ times, we can then implement an arbitrary transformation $\exp(-i H)$ on the system. As the error per step {is} $\mathcal{O}\left(\frac{1}{N^2}\right)$ and there are $\mathcal{O}\left(N\right)$ steps, the overall error can be made as small as desired by choosing $N$ sufficiently large. 

\bigskip

Next, let us see how the the conserved quantities are stored in the batteries (i.e.,~the reference frame). For simplicity, we again consider implementing the small transformation $U_{\alpha, 0}$ generated by $O_0$ on the system.  The change in the  conserved quantity $O_{k}$ for particle $r$ in the reference frame is given by 
\begin{align*}
\Delta O_{k}^r = \trace{{O}_k^r \, V \, \rho_S \otimes \rho_R \, V^\dagger } - \trace{O_{k} \, \rho^{(r)}}\,,
\end{align*}
where ${O}_k^r := \id \otimes ... \otimes O_{k} \otimes \ldots \otimes \id$ is the operator that is $\id$ everywhere except for the $r$-th frame particle where it is $O_{k}$.
Expanding $V$ to first order in $\frac{1}{N}$ we obtain 
\begin{align} 
\Delta O_{k}^r &=-i \frac{1}{\eta_1 \ldots \eta_D} \frac{\alpha}{N} \, \trace{{O}_k^r\, [T,  \rho_S \otimes \rho_R ] } + \mathcal{O}\left(\frac{1}{N^2}\right) \nonumber \\
& = -i \frac{1}{\eta_1 \ldots \eta_D} \frac{\alpha}{N}\, \sum_{a, a_1, \ldots, a_D} f_{aa_1\ldots a_D} \trace{O_{k} \,[O_{a_r}, \rho^{(r)}] } \trace{O_a \rho_S} \prod_{j \neq r} \trace{O_{a_j} \rho^{(j)}} + \mathcal{O}\left(\frac{1}{N^2}\right) \nonumber \\
&= -i \, \frac{1}{ \eta_r} \frac{\alpha}{N}\, \sum_{a,a_r} f_{a,1,\ldots,r-1, a_r,r+1,D} \, \trace{O_a \rho_S} \, \trace{O_{k}\, [O_{a_r}, \rho^{(r)}] }  + \mathcal{O}\left(\frac{1}{N^2}\right) \,. \label{eq:elfra}
\end{align}
Now notice that the sum in {Eq.~\eqref{eq:elfra}} consists of two terms: $\{a=0,a_r=r\}$ and $\{a=r,a_r=0\}$, since any other assignment of values to $a$ or $a_r$ will render $f_{a,1,\ldots,r-1, a_r,r+1,D} = 0$. Hence, 
\begin{align}
\Delta O_{k}^r =   -i \, \frac{1}{ \eta_r} \frac{\alpha}{N}\, \Big( &f_{0,1,\ldots,r-1, r,r+1,D} \, \trace{O_0 \rho_S} \, \trace{O_{k}\, [O_{r}, \rho^{(r)}]}   \\
&+ f_{r,1,\ldots,r-1,0,r+1,D} \, \trace{O_r \rho_S} \, \trace{O_{k}\, [O_0, \rho^{(r)}]} \Big) + \mathcal{O}\left(\frac{1}{N^2}\right) .
\end{align}
Using $\rho^{(r)} = \frac{1}{d} \left( \id + \frac{1}{\| O_k \|} \, O_k \right)$ and  $f_{r,1,\ldots,r-1,0,r+1,D}=-1$, we obtain  
\begin{align}
\Delta O_{k}^r =   i \,  \frac{\alpha}{N} \frac{\trace{O_r \rho_S} \, \trace{O_{k} \, [O_0,  O_r]}}{\tr{O_r^2}{}}+ \mathcal{O}\left(\frac{1}{N^2}\right) .
\end{align}
So far, we have only used the first two properties of the operators $O_k$. Given the third property that the operators are closed under commutation, we find that either $[O_0,O_r]=0$, in which case the $r^{th}$ subsystem in the reference frame accumulates no changes in any conserved quantity (i.e., $\Delta O_{k}^r = 0$ for all $k$), or $[O_0,O_r] \propto O_m$. In {the latter,} the $r^{th}$ subsystem will only accumulate changes in the conserved quantity $O_m$ (i.e., $\Delta O_{k}^r = 0$ unless $k=m$). As each subsystem accumulates changes in at most one conserved quantity, it is possible to separate all of the conserved quantities into different batteries. A similar result will apply for small rotations generated by any $O_k$ and thus for the overall transformation.

\subsection{$T$ preserves extended conserved quantities}
\label{ap:b21}
In this subsection, we show that $T$ commutes with all extensive conserved quantities. Let us begin by revisiting the case with $d=2$, where our conserved quantities are $s_x$, $s_y$ and $s_z$, as presented in {the main text for spin-$\frac12$ systems.} Here, T acts on three particles as 
\begin{align*}
T = \sum_{a, a_1, a_2} f_{aa_1a_2} \, s_a \otimes s_{a_1} \otimes s_{a_2}\,.
\end{align*}
Now consider the conserved quantity corresponding to the total angular momentum in the $x$-direction, $S_x= s_x \otimes \id \otimes \id + \id \otimes s_x \otimes \id + \id \otimes \id \otimes s_x$, which is the sum of the angular momentum of the three individual particles in the $x$-direction.
The commutator $[S_x, T]$ may be expressed as
\begin{align} \nonumber
[S_x, T] =&  \sum_{a, a_1, a_2} f_{aa_1a_2} \, [s_x, s_a] \otimes s_{a_1} \otimes s_{a_2} \\
&+  \sum_{a, a_1, a_2} f_{aa_1a_2} \,  s_a \otimes [s_x, s_{a_1}] \otimes s_{a_2} \\
&+  \sum_{a, a_1, a_2} f_{aa_1a_2} \, s_a \otimes s_{a_1} \otimes [s_x, s_{a_2}] \,. \nonumber
\end{align}
We will now see how each term in each sum is either zero or cancelled out by a similar term appearing in a different sum but with the opposite sign. {First, note} that all terms in the final answer must contain the same spin operator on exactly two particles and a different spin operator on the third particle. This is because $a, a_1$, and $ a_2$ must all be distinct in order for $f_{aa_1a_2}$ to be non-zero, and the spin operator generated by the commutator is different from the spin operators appearing within it. As an example, a term proportional to $s_y \otimes s_y \otimes s_x$ can be generated by   $[s_x, s_z] \otimes s_y \otimes s_x$ and   $s_y \otimes [s_x, s_z] \otimes s_x$. The former  of these will contribute to $[S_x, T]$ with a coefficient $f_{zyx}$ via the first sum, whereas the latter will contribute with a coefficient $f_{yzx}$ via the second sum. As $f_{zyx}=-f_{yzx}$ these two terms will cancel out, and the same applies to all other terms, giving a total commutator of zero.

\bigskip
Now let us see how a similar reasoning applies to the higher dimensional case.  We want to show that {the total conserved quantity $O_k^{\mathrm{tot}} = O_k + O_k^{R}$ is preserved by $T$, i.e., that $[O_k^{\mathrm{tot}}, T]=0$, where}
\begin{align*}
[O_k^{\mathrm{tot}}, T] = \sum_{r=0}^n [{O}_k^r, T] \,  
\end{align*}
with
\begin{align*}
[{O}_k^r, T] = \sum_{a, a_1, \ldots, a_D} f_{aa_1\ldots a_D} \, O_a \otimes \ldots \otimes [O_k, O_{a_r}] \otimes \ldots \otimes \ldots \otimes O_{a_D}\,.
\end{align*}

We will see how each term in $[{O}_k^r, T] $, for every $r$, is either $0$ or cancelled out by a term in $[{O}_k^u, T] $, for some {other} $u$. To begin, take a fixed but arbitrary $r$, and a non-zero term in $[{O}_k^r, T] $. This term is identified by its indices $aa_1\ldots a_D$. Given the `closed under commutation' property of the operators in $\cM$, $[O_k, O_{a_r}] = \xi \, O_{b}$, where $b \in \{a, a_1, \ldots, a_D\} \setminus \{a_r\} $, and $\xi$ is a constant that may depend on $k$ and $a_r$. Let $u$ be such that $a_u = b$, and for simplicity, and without loss of generality, assume $u > r$. Hence, 
\begin{align*}
O_a \otimes \ldots \otimes [O_k, O_{a_r}] \otimes \ldots \otimes O_{a_u}  \otimes \ldots \otimes O_{a_D} &= \xi \, O_a \otimes \ldots \otimes O_b \otimes \ldots \otimes O_{b} \otimes \ldots \otimes  O_{a_D} \\
&= O_a \otimes \ldots \otimes O_{a_u} \otimes \ldots \otimes [O_k, O_{a_r}] \otimes \ldots  \otimes O_{a_D} \,.
\end{align*}
Now, the term with indices $(aa_1 \ldots a_r \ldots a_u \ldots a_D)$ will contribute to $[O_k^{\mathrm{tot}}, T] $ via the sum in the term $[{O}_k^r, T] $ with coefficient $f_{aa_1 \ldots a_r \ldots a_u \ldots a_D}$. The same term will however appear in the sum in the term $[{O}_k^u, T] $, with indices $(aa_1 \ldots a_u \ldots a_r \ldots a_D)$, and will contribute to $[\cO_k, T] $ with coefficient $f_{aa_1 \ldots a_u \ldots a_r \ldots a_D}$. By the properties of $f$, $f_{aa_1 \ldots a_r \ldots a_u \ldots a_D} = - f_{aa_1 \ldots a_u \ldots a_r \ldots a_D}$, and the two terms cancel out. As every term in the final answer can be generated in exactly two ways with opposite coefficients, the result that $[O_k^{\mathrm{tot}}, T]=0$ then follows. Although we will think of the operators $O_k^{\mathrm{tot}}$ as the conserved quantities of interest, note that any other extensive conserved quantity can be expressed as a linear combination of the operators $O_k^{\mathrm{tot}}$ and the identity, this means that $T$ commutes with all extensive conserved quantities. 

\medskip
Interestingly, one can also show that $[O_k^{\mathrm{tot}}, T]=0$ without the `closed under commutation'  property, by using the fact that  $[O_k, O_\ell] = \sum_{m\neq k,\ell} w_{k\ell m} O_m$ for some coefficients $w_{k\ell m}$. As above, every non-zero term in $[O_k^{\mathrm{tot}}, T]=0$ can be generated in exactly two ways with opposite coefficients. Hence for any complete set of traceless orthonormal operators $O_k$, the operator $T$ constructed from those operators commutes with all  extensive conserved quantities.

\subsection{Dimension $2^n$}
\label{ap:2n}
Let us consider the case where we have a system $S$ of dimension $d=2^n$. Here, we can indeed find an operator basis with the properties specified above, and can thus separate any changes to a complete basis of conserved quantities. In particular, we take the set $ \cM$ of $K=4^n-1$ operators  $\{O_k\}$ to consist of products of spin-$\frac{1}{2}$ operators. That is, each $O\in \cM$ has the form $O=s_1 \otimes \ldots \otimes s_n$, where $s_k \in \left\{\frac{1}{2}\id, \frac{1}{2}\sigma_x, \frac{1}{2}\sigma_y, \frac{1}{2}\sigma_z\right\}$, but where not all $s_k$ may equal the identity. 

It is easy to verify that these operators are traceless and orthogonal. Furthermore, the product of any two distinct operators in $\cM$ is proportional to another operator in the set, i.e., there exists an $m$ such that  $O_k O_\ell \propto O_m$.  By taking the Hermitian conjugate, we also have  $O_\ell O_k \propto O_m$. Hence either $[O_k, O_\ell]=0$, or $[O_k, O_\ell]\propto O_m$ and these operators are hence closed under commutation as defined above.

\subsection{{Explicit} bounds} 

As in the previous case, we can also calculate explicit bounds on the $\cO\left( \frac{1}{N^2} \right)$  terms in the bounds above, using similar techniques to those in {Appendix} \ref{ap:spindec} and  in  Ref.~\cite{QRF18}.

In particular, we find 
\begin{align} 
\left\| \, \tr{V \, \rho_S \otimes \rho_R \, V^\dagger}{R} - (\rho_S -i \frac{\alpha}{N} \, [O_0,\rho_S]) \, \right\|_1 & \leq  \left( 2 K!\, \frac{\|O_0\|\|O_1\| \ldots \|O_D\|}{\eta_1 \eta_2 \ldots \eta_D} \right)^2  \,(e-2)\,\left( \frac{\alpha}{N}\right)^2\,, \label{eq:31}
\end{align}
where we have used an approach similar to that in the derivation of {Eq.~\eqref{eq:16},}  assuming that $N> \left( 2 K!\, \frac{\|O_0\|\|O_1\| \ldots \|O_D\|}{\eta_1 \eta_2 \ldots \eta_D} \right) \alpha$. 

Applying now {Eq.~(D5)} of Ref.~\cite{QRF18} to this setup yields
\begin{equation}\label{eq:32} 
\left\| U_{\alpha, 0}  \rho_S U_{\alpha, 0} ^{\dagger} - \left(\rho_S -i \frac{\alpha}{N} \, [O_0,\rho_S] \right) \right\|_1 \leq 4 (e-2) \left( \frac{\alpha}{N} \right)^2,
\end{equation} 

Combining {Eqs.~\eqref{eq:31}} and \eqref{eq:32} through the triangle inequality, we find
\begin{align} \label{eq:33}
\left\|\tr{V \, \rho_S \otimes \rho_R \, V^\dagger}{R} - U_{\alpha, 0}  \rho_S U_{\alpha, 0} ^{\dagger} \right\|_1 &\leq 4(e-2)\left( 1+  \left(  K!\, \frac{\|O_0\|\|O_1\| \ldots \|O_D\|}{\eta_1 \eta_2 \ldots \eta_D} \right)^2 \right) \, \left( \frac{\alpha}{N} \right)^2\,.
\end{align}
Finally, following a similar approach we obtain
\begin{align}
\left| \Delta O_{k}^r -  i \,  \frac{\alpha}{N} \frac{\trace{O_r \rho_S} \, \trace{O_{k} \, [O_0,  O_r]}}{\tr{O_r^2}{}} \right|   \leq \| O_k\| \left( 2 K!\, \frac{\|O_0\|\|O_1\| \ldots \|O_D\|}{\eta_1 \eta_2 \ldots \eta_D} \right)^2  \,(e-2)\,\left( \frac{\alpha}{N}\right)^2. 
\end{align}
In the particular case considered above of dimension $d=2^n$, in which the operators $O_k$ are products of spin-\half operators and $\id/2$, note that $\eta_k = \| O_k\| =\frac{1}{d}$ and hence these expressions simplify considerably.

\end{document}